\documentclass[12pt]{article}
\usepackage[margin=1in,left=1in]{geometry}

\usepackage{amsthm}
\usepackage{amssymb}
\usepackage{amsmath}
\usepackage{mathtools}\usepackage{adjustbox}
\usepackage{xcolor}
\usepackage{stmaryrd}    
\usepackage[all]{xy}
\usepackage{rotating}
\usepackage{xfrac}
\usepackage{enumitem}
\setlist{nosep} 

\usepackage{hhline}

\usepackage{tikz}
\usetikzlibrary{
  backgrounds,
  decorations.pathmorphing,
  decorations.markings,
  decorations.text
}
\usetikzlibrary{cd}

\usepackage{mathrsfs} 
\DeclareMathAlphabet{\mathpzc}{OT1}{pzc}{m}{it} 

\usepackage{hyperref}
\hypersetup{
  colorlinks = true,
  allcolors=.
}

\let\PLAINthebibliography\thebibliography
\renewcommand\thebibliography[1]{
  \PLAINthebibliography{#1}
  \setlength{\parskip}{0.5pt}
  \setlength{\itemsep}{0.5pt plus .3ex}
}

\usepackage[titletoc]{appendix}

\usepackage[safe]{tipa} 

\definecolor{nyulight}{RGB}{107, 33, 158}

\definecolor{darkblue}{rgb}{0.05,0.25,0.65}
\definecolor{darkgreen}{RGB}{20,140,10}
\definecolor{lightgray}{rgb}{0.9,0.9,0.9}
\definecolor{darkorange}{RGB}{200,100,5}
\definecolor{darkyellow}{rgb}{.91,.91,0}

\usepackage{multirow}

\usepackage{enumerate} 

\newcommand{\defneq}{\equiv}

\newcommand{\closed}{\mathrm{clsd}}

\newcommand{\Differential}{\mathrm{d}}
\newcommand{\differential}{\Differential}

\usepackage[new]{old-arrows}   

\newdir{> }{{}*!/10pt/@{>}}

\usepackage{amssymb}

\DeclareRobustCommand{\rchi}{{\mathpalette\irchi\relax}}
\newcommand{\irchi}[2]{\raisebox{\depth}{$#1\chi$}} 

\makeatletter
\newif\if@sup
\newtoks\@sups
\def\append@sup#1{\edef\act{\noexpand\@sups={\the\@sups #1}}\act}%
\def\reset@sup{\@supfalse\@sups={}}%
\def\mk@scripts#1#2{\if #2/ \if@sup ^{\the\@sups}\fi \else%
  \ifx #1_ \if@sup ^{\the\@sups}\reset@sup \fi {}_{#2}%
  \else \append@sup#2 \@suptrue \fi%
  \expandafter\mk@scripts\fi}
\def\tensor#1#2{\reset@sup#1\mk@scripts#2_/}
\def\multiscripts#1#2#3{\reset@sup{}\mk@scripts#1_/#2%
  \reset@sup\mk@scripts#3_/}
\makeatother

\makeatletter
\newbox\slashbox \setbox\slashbox=\hbox{$/$}
\def\itex@pslash#1{\setbox\@tempboxa=\hbox{$#1$}
  \@tempdima=0.5\wd\slashbox \advance\@tempdima 0.5\wd\@tempboxa
  \copy\slashbox \kern-\@tempdima \box\@tempboxa}
\def\slash{\protect\itex@pslash}
\makeatother

\def\clap#1{\hbox to 0pt{\hss#1\hss}}
\def\mathllap{\mathpalette\mathllapinternal}
\def\mathrlap{\mathpalette\mathrlapinternal}
\def\mathclap{\mathpalette\mathclapinternal}
\def\mathllapinternal#1#2{\llap{$\mathsurround=0pt#1{#2}$}}
\def\mathrlapinternal#1#2{\rlap{$\mathsurround=0pt#1{#2}$}}
\def\mathclapinternal#1#2{\clap{$\mathsurround=0pt#1{#2}$}}

\DeclareSymbolFont{symbolsC}{U}{txsyc}{m}{n}
\SetSymbolFont{symbolsC}{bold}{U}{txsyc}{bx}{n}
\DeclareFontSubstitution{U}{txsyc}{m}{n}

\usepackage{cleveref}

\crefformat{section}{\S#2#1#3} 
\crefformat{subsection}{\S#2#1#3}
\crefformat{subsubsection}{\S#2#1#3}

\newtheorem{theorem}{Theorem}[section]
\newtheorem{lemma}[theorem]{Lemma}

\theoremstyle{definition}
\newtheorem{definition}[theorem]{Definition}

\newtheorem{remark}[theorem]{Remark}

\usepackage{amsfonts}
\usepackage{colortbl}

\renewcommand{\emph}{\textit}

\begin{document}

\setlength{\abovedisplayskip}{3pt}
\setlength{\belowdisplayskip}{3pt}
\setlength{\abovedisplayshortskip}{-10pt}
\setlength{\belowdisplayshortskip}{3pt}

\title{Flux Quantization on Phase Space}

\author{
  Hisham Sati${}^{\ast \dagger}$
  \;\;
  and
  \;\;
  Urs Schreiber${}^{\ast}$
}

\maketitle

\thispagestyle{empty}

\begin{abstract}
 While it has become widely appreciated that (higher) gauge theories need, besides their variational phase space data, to be equipped with
 ``flux quantization laws'' in generalized differential cohomology, there used to be no general prescription for how to define and construct
 the resulting flux-quantized phase space stacks.

 \vspace{1mm}
 In this short note we observe that all higher Maxwell-type equations
 have solution spaces given by flux densities on a Cauchy surface subject to a higher Gau{\ss}
 law and no further constraint: The metric duality-constraint is all absorbed into the evolution equation away from the Cauchy surface.

 \vspace{1mm}
 Moreover, we observe that the higher Gau{\ss} law characterizes the Cauchy data as flat differential forms valued
 in a characteristic $L_\infty$-algebra. Using the recent construction of the non-abelian Chern-Dold character map, this implies that
 compatible flux quantization laws on phase space have classifying spaces whose rational Whitehead $L_\infty$-algebra is this characteristic one.
 The flux-quantized higher phase space stack of the theory is then simply the corresponding (generally non-abelian) differential cohomology
 moduli stack on the Cauchy surface.

 \vspace{1mm}
 We show how this systematic prescription reproduces existing proposals for flux-quantized phase spaces of vacuum Maxwell theory and
 of the chiral boson and its higher siblings, but reveals that there are other choices of (non-abelian) flux quantization laws even
 in these basic cases, further discussed in a companion article \cite{SS23QuantumFluxObservables}.

 \vspace{1mm}
 Moreover, for the case of NS/RR-fields in type II supergravity/string theory, the traditional ``Hypothesis K'' of flux quantization in
 topological K-theory is naturally implied, without the need, on phase space, of the notorious further duality constraint.

 \vspace{1mm}
 Finally, as a genuinely non-abelian example we consider flux-quantization of the C-field in 11d supergravity/M-theory given by
 unstable differential 4-Cohomotopy (``Hypothesis H'') and emphasize again that, implemented on Cauchy data, this qualifies as the full phase
 space without the need for a further duality constraint.

\end{abstract}

\vspace{.3cm}

\vfill

\hrule
\vspace{5pt}

{
\footnotesize
\noindent
\def\arraystretch{1}
\tabcolsep=0pt
\begin{tabular}{ll}
${}^*$\,
&
Mathematics, Division of Science; and
\\
&
Center for Quantum and Topological Systems,
\\
&
NYUAD Research Institute,
\\
&
New York University Abu Dhabi, UAE.
\end{tabular}
\hfill
\adjustbox{raise=-15pt}{
\includegraphics[width=3cm]{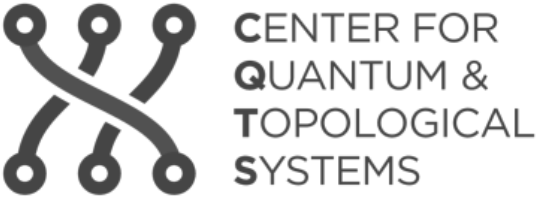}
}
}

\vspace{1mm}
\noindent
{
\footnotesize
${}^\dagger$The Courant Institute for Mathematical Sciences, NYU, NY
}

\vspace{.2cm}

\noindent
\scalebox{.82}{
The authors acknowledge the support by {\it Tamkeen} under the
{\it NYU Abu Dhabi Research Institute grant} {\tt CG008}.
}

\newpage

\section{Introduction}

\noindent
{\bf The need of flux quantization.}
It has become widely appreciated that, besides their equations of motion, (higher) gauge theories are to be subjected to ``flux quantization laws''
(we follow \cite[Introd.]{Char}, survey in \cite{SS24Flux}, see also \cite{Freed00}), even on the classical fields. Here ``quantization'' is in the sense of ``discretization'': For abelian gauge
theories the flux quantization restricts the total fluxes through --- and hence the charges inside  --- closed submanifolds of spacetime to lie
on a lattice (``charge lattice''), hence to be integral multiples of certain indecomposable units. The archetypical case is the quantization
of electromagnetic flux, which Dirac originally discussed (cf. \cref{MaxwellField}, whence one also speaks of ``Dirac charge quantization'')
for hypothetical magnetic monopoles, but which also controls the experimentally observed quantization of magnetic flux through type II
superconductors witnessed by integer numbers of ``Abrikosov vortex strings'' (cf. \cite[\S 2.1]{SS24Flux}).

\medskip

\noindent
{\bf Non-perturbative data in flux quantization.}
In fact, flux quantization is more than just a condition on gauge fields, it is also ``extra structure'' in mathematical jargon, which
physically means: it involves adjoining {\it further fields} or at least further field components such as ``torsion'' components, which
are invisible to traditional perturbation theory. In the basic example of electromagnetism, Dirac flux quantization promotes the gauge
potential 1-form $A$ to a connection on a principal bundle, equivalently to a cocycle in {\it integral differential cohomology} (cf. Rem. \ref{ScopeOfExamples}).
The extra data involved in flux quantization encodes non-perturbative solitonic degrees of freedom (such as Dirac monopoles or Abrikosov vortices).

\medskip
More generally but in the absence of non-linear self-sourcing of the fluxes, their quantization is in abelian Whitehead-generalized differential
cohomology (review in \cite{Bunke12}\cite{Szabo13} \cite{AmabelDebrayHaine21}\cite{Debray24}). If there are non-linear interactions (such as for the
C-field in 11d supergravity, \cref{TheCField}) the flux qantization must be in a non-abelian differential cohomology \cite{Char}\footnote{Where below we give
equation numbers etc. for \cite{Char},  we are referring to the published version,
see \href{https://ncatlab.org/schreiber/show/The+Character+Map\#PublishedVersion}{\tt ncatlab.org/schreiber/show/The+Character+Map\#PublishedVersion},
which differs from the numbering in the arXiv version (otherwise the content is the same).}. Early proposals in the topological case were given in
\cite{Sati05}\cite{Sati06}.

\medskip
\noindent
{\bf The open problem of flux-quantized phase spaces.}
However, this extra nonperturbative data involved in the flux quantization of higher gauge theories is largely outside the scope of traditional recipes
for constructing physical theories, certainly beyond what may be derived by variational calculus from a Lagrangian density. For one, the phase space
of a flux-quantized gauge theory is no longer a manifold, not even a dg-manifold as in BRST-BV formalism, but must be a {\it higher stack}
(a {\it smooth $\infty$-groupoid} \cite[\S 3.1]{SSS13}\cite{SS20Orbi}\cite[pp. 41]{Char}, exposition in \cite{FSS14Stacky}\cite{Schreiber24}).
Therefore, previous discussions of flux-quantized phase spaces have been rare and somewhat ad-hoc.

\medskip
\noindent
{\bf A systematic construction: Higher Gau{\ss} laws in non-abelian cohomology.}
Here we discuss a systematic procedure for arriving at the flux-quantized phase space ($\infty$-stack) of a given higher gauge theory, with an emphasis on
some aspects that have received little to no attention before: The choices involved in the process, and its generalization to ``non-abelian''
gauge theories with non-linear self-interactions, such as exhibited by the C-field in 11d supergravity (\cref{TheCField}).

Our approach is non-perturbative and does not use any Lagrangian density (in particular it applies to ``non-Lagrangian'' field theories for
which a Lagrangian density does not or not naturally exist, such as self-dual higher gauge theories, cf. \cref{ChiralBoson} \& \cref{RRFields});
the input instead is the equations of motion on the flux densities in the
form of ``higher Maxwell equations'' \eqref{TheCovariantEquationsOfMotion}:

\smallskip
{\bf Covariant and canonical phase spaces of pure higher gauge theories.}
To appreciate this, recall that the {\it phase space} of a field theory is, in generality, really the {\it space of on-shell field histories},
as such also known as the {\it covariant phase space}, for emphasis (\cite[p. 314]{Witten86}\cite{CrnkoviWitten87}\cite[\S 17.1]{HenneauxTeitelboim92}; see \cite{Khavkine14}\cite{GiotopoulosSati23} for rigorous discussion).
The traditional discussion of phase spaces by Cauchy data with respect to a given spatial foliation of spacetime  (``canonical'' phase spaces),
and particularly as cotangent bundles, is just one way (applicable in good situations) to identify on-shell field histories with suitable initial value data,
in particular with initial field ``coordinates'' and their ``canonical momenta''.

\smallskip

Now in our case of pure higher gauge theories, we may observe that the on-shell field histories are those (flux-quantized) gauge potentials whose
flux densities (only) are subjected to the corresponding higher Maxwell equations \eqref{TheCovariantEquationsOfMotion}. Conversely, this means that
the phase space of a higher gauge theory must be the partial extension by all compatible gauge potentials of the {\it solution space} of flux
densities satisfying the higher Maxwell equations:
\begin{equation}
  \label{PhaseSpaceOverSolutionSpace}
  \hspace{-2.5cm}
  \begin{tikzcd}[
    row sep=1,
    column sep=40pt
  ]
    &&&[10pt]
    \mathrm{PhaseSpace}
    \ar[
      d,
      "{
        \scalebox{.7}{
          \color{darkgreen}
          \bf
          extract fluxes
        }
      }"{description, pos=.42}
    ]
    \ar[
      r,
      equals,
      "{
        \scalebox{.7}{
        Def. \ref{IntrinsicPhaseSpace}
        }
      }"{swap, yshift=-2pt}
    ]
    &
    \overset{
      \mathclap{
        \raisebox{1pt}{
          \scalebox{.7}{
            \color{darkblue}
            \bf
            Flux densities
          }
        }
      }
    }{
    \mathbf{\Omega}_{\mathrm{dR}}\big(
      X^d
      ;\,
      \mathfrak{l}\mathcal{A}
    \big)_\closed
    }
    \;\;
    \underset{
      \mathclap{
        \underset{
          \mathclap{
            \scalebox{.7}{
              \color{darkorange}
              \bf
              \def\arraystretch{.9}
              \begin{tabular}{c}
                stacky homotopy
                \\
                fiber product
              \end{tabular}
            }
          }
        }{
        \raisebox{-4pt}{
          \scalebox{.7}{$
              L^{\mathbb{R}}\mathcal{A}(X^d)
          $}
        }
        }
        \hspace{-10pt}
      }
    }{\times}
    \quad
    \overset{
      \mathclap{
        \raisebox{1pt}{
          \scalebox{.7}{
            \color{darkblue}
            \bf
            Flux quantization law
          }
        }
      }
    }{
      \mathcal{A}(X^d)
    }
    \ar[
      d,
      shift right=20pt,
      shorten <=-34pt,
      "{
        \scalebox{.7}{
          \color{darkgreen}
          \bf
          projection
        }
      }"{description, pos=-30}
    ]
    \hspace{-40pt}
    \\
    \mathbb{R}^p
    \ar[
      rrr,
      "{ \vec F }"{description},
      "{\hspace{-5mm}
        \scalebox{.7}{
          \color{darkblue}
          \bf
          \def\arraystretch{.9}
          \begin{tabular}{c}
            Flux densities solving
            \\
            higher Maxwell equations
          \end{tabular}
        }
      }"{swap, yshift=-4pt}
    ]
    \ar[
      urrr,
      "{ \vec A }"{description},
      "{
        \scalebox{.65}{
          \color{darkblue}
          \bf
          \def\arraystretch{.9}
          \begin{tabular}{c}
            Flux-quantized
            \\
            gauge potentials
          \end{tabular}
        }
      }"{sloped, yshift=2pt, pos=.4}
    ]
    &&&
    \underset{
      \mathclap{
        \raisebox{-5pt}{
          \scalebox{.7}{
            \eqref{SolutionSpace}
          }
        }
      }
    }{
      \mathrm{SolSpace}
    }
    \ar[
      r,
      equals,
      "{
        \scalebox{.7}{
          Thm. \ref{PhaseSpaceViaGaussLaws}
        }
      }"{swap, yshift=-2pt}
    ]
    &
    \;\;
    \underset{
      \mathclap{
        \raisebox{.3pt}{
          \scalebox{.7}{
            \color{darkblue}
            \bf
            \def\arraystretch{.9}
            \begin{tabular}{c}
              Cauchy data subject to
              \\
              $L_\infty$-algebraic Gau{\ss} law
            \end{tabular}
          }
        }
      }
    }{
    \mathbf{\Omega}_{\mathrm{dR}}\big(
      X^d
      ;\,
      \mathfrak{l}\mathcal{A}
    \big)_\closed
    }
  \end{tikzcd}
\end{equation}

\noindent
Indicated in the diagram on the right are our main observations in the following \cref{FluxQuantizationOnPhaseSpace}:
\begin{itemize}
\item[{\bf (i)}] {\bf Thm. \ref{PhaseSpaceViaGaussLaws}}: The solution spaces of flux densities of common higher gauge theories (such as appearing in higher
supergravity theories) are identified with spaces of flat $\mathfrak{a}$-valued differential forms, encoding the higher Gau{\ss}
law constraint on any Cauchy surface \eqref{FoliationByCauchySurfaces}, for $\mathfrak{a}$ a connective $L_\infty$-algebra of
finite type (Rem. \ref{ModuliProblemOfSolutions}).

In particular, this means that the phase space is given by purely cohomological data, while the metric data (Hodge-duality on fluxes) is all absorbed into
the isomorphism which identifies points in phase space with on-shell field histories.

\item[{\bf (ii)}] {\bf Def. \ref{IntrinsicPhaseSpace}}: A natural notion of flux quantization then is a choice of classifying space $\mathcal{A}$ whose
rational Whitehead $L_\infty$-algebra is $\mathfrak{l}\mathcal{A} \,\simeq\, \mathfrak{a}$, in which case the corresponding phase space of on-shell gauge potentials
is the moduli stack of (possibly non-abelian) differential $\mathcal{A}$-cocycles on the Cauchy surface, following \cite{Char}.

\end{itemize}

\medskip

\noindent
{\bf Examples and Applications.}
In \cref{ExamplesAndApplications} we first note that for basic cases like ordinary Maxwell theory (\cref{MaxwellField}) and the chiral boson (\cref{ChiralBoson})
this prescription subsumes existing proposals for flux-quantized phase spaces (but we highlight that there are other consistent non-abelian flux quantizations
even for these abelian theories).

The case of ordinary Maxwell theory may serve to appreciate the general result: Here it is well-familiar that {\it on a Cauchy surface} the magnetic flux
density $\differential A$ is {\it independent} from the electric flux density $E$; in fact, the latter is the canonical momentum to the canonical coordinate $A$.
It is only when extending initial value data $(\differential A, E)$ uniquely from the Cauchy surface to an on-shell field history on all of space-time that $E$
becomes identified with the spacetime-Hodge dual of $\differential A$. This Hodge duality relation, instead of being a constraint on the initial value data,
is part of the rule for extending that data to all of spacetime.

\medskip

Accordingly, in the example of the RR-fields in type II supergravity (\cref{RRFields}) our prescription gives flux quantization in topological K-theory for fields
{\it on a Cauchy surface} and here {\it without} a further self-duality constraint on the K-theory. This matches (and hence justifies) the general practice of K-theoretic
computations of D-brane charges.

\smallskip

In the same vein, it follows (in \cref{TheCField}) that a flux-quantized phase space of the C-field of 11d supergravity is given by the moduli
stack of differential Cohomotopy on a Cauchy surface, again without any further duality constraint. We close by highlighting some implications of this result.

\medskip
\medskip

\noindent
{\bf Acknowledgment.} We thank Grigorios Giotopoulos for useful discussion and the anonymous referees for helpful comments.


\section{Flux quantization on Phase space}
\label{FluxQuantizationOnPhaseSpace}

{\bf In \cref{HigherFluxDensities}} we prove that the space of on-shell flux densities of any higher gauge theory controlled by ``higher Maxwell equations'' on
a globally hyperbolic spacetime is equivalently the space of flux densities on a Cauchy surface satisfying a higher Gau{\ss} law (and otherwise unconstrained),
which in turn is equivalently a space of flat differential forms on the Cauchy surface with coefficients in a characteristic $L_\infty$-algebra.

\smallskip

\noindent
{\bf In \cref{FluxQuantizationAndPhaseSpace}} we observe that, therefore, a flux-quantized phase space for such a theory is given by the canonical differential
$\mathcal{A}$-cohomology for any classifying space $\mathcal{A}$ whose rational Whitehead $L_\infty$-algebra is the characteristic one.

Here, by ``phase space'', we mean the underlying space (smooth $\infty$-stack) while disregarding, for this note, its Poisson brackets
and their quantization. That also the algebras of quantum observables on fluxes have an immediate homotopy-theoretic construction, at least in the topological
sector, is the topic of the companion article \cite{SS23QuantumFluxObservables}.

\subsection{Higher Maxwell equations on Higher flux densities}
\label{HigherFluxDensities}

That the main result we are presenting here (Thm. \ref{PhaseSpaceViaGaussLaws} below) has a rather easy proof (Lem. \ref{GaussLawIsFirstClassConstraint}, cf. \cref{Computations})
may be attributed to the well-adapted form \eqref{TheCovariantEquationsOfMotion} into which we cast the general higher Maxwell equations to start with,
known in electromagnetism (cf. \cref{MaxwellField}) as the ``premetric'' form \cite{HehlItinObukhov16} and in supergravity and string theory as
the ``duality symmetric'' or ``democratic'' form (cf. references in \cref{RRFields} \& \cref{TheCField}).

\smallskip
In a sense, this duality-symmetric formulation makes {\it all} higher gauge theories appear as self-dual higher gauge theories
(with ``doubled'' field content if they are not self-dual in the na{\"i}ve sense, as in the example \cref{MaxwellField}); which conversely
means that our formulation of self-dual higher gauge theories (as in Ex. \cref{ChiralBoson}, \cref{RRFields}) is no more intricate than the
general case, in stark contrast to traditional approaches.

\newpage

\noindent
{\bf Globally hyperbolic spacetime.}
To set the scene for canonical phase space analysis, consider the following:

\begin{itemize}[leftmargin=.5cm]
\item
$X^D$ a $D$-dimensional spacetime manifold,
\item with corresponding Hodge star operator on differential forms (eg. \cite[\S 14.1a]{Frankel97}):
\vspace{1mm}
\begin{equation}
  \label{HodgeStarOperator}
 \star
 \,:\,
 \Omega^{p}_{\mathrm{dR}}\big(
   X^{D}
 \big)
 \xrightarrow{\;\; \sim \;\;}
 \Omega^{D-p}_{\mathrm{dR}}\big(
   X^{D}
 \big)
 ,
 \;\;\;\;
  \star \, \star
  \;=\;
  -(-1)^{p(D-1)}
  \,,
\end{equation}
\item
which is globally hyperbolic (e.g. \cite[\S 3.11]{MinguzziSanchez08}) exhibited by a smooth foliation by spacelike Cauchy surfaces (cf. \cite[Thm. 1.1]{BernalSanchez05}):
\begin{equation}
  \label{FoliationByCauchySurfaces}
  X^{D}
    \;\simeq\;
  \mathbb{R}^{0,1} \times X^{d}
  \,,
\end{equation}
\item $X^d \hookrightarrow X^D $
its $d$-dimensional Cauchy surface
at $t = 0$;
\item
$\partial_t \,\in\, \Gamma (T X^{D})$ the corresponding timelike vector field,
\item
with its contraction operation
$$
  \iota_{\partial_t}
  \;:\;
  \Omega^\bullet_{\mathrm{dR}}(X^{D})
  \longrightarrow
  \Omega^{\bullet-1}_{\mathrm{dR}}(X^{D})
  \,,
$$
\item
whose kernel we denote:
$$
  \Omega^\bullet_{\mathrm{dR}}(X^{D})_{\iota_{\partial t} = 0}
  \;\defneq\;
  \mathrm{ker}(\iota_{\partial_t})
  \,.
$$
\item
This induces a decomposition of differential forms
\begin{equation}
  \label{DecompositionOfDifferentialForms}
  \Omega^{p}_{\mathrm{dR}}(X^{D})
  \;\simeq\;
  \Omega^p_{\mathrm{dR}}(X^{D})_{\iota_{\partial t} = 0}
  \;\oplus\;
  \star
  \,
  \Omega^{D-p}_{\mathrm{dR}}(X^{D})_{\iota_{\partial t} = 0}
\end{equation}
into summands that have none or have one wedge factor of $\differential t$, respectively.
\item
The de Rham differential on $X^{D}$ accordingly decomposes into a temporal and a spatial summand
$
  \differential
  \,=\,
  \differential_t
  \,+\,
  \differential_{s}
$
$$
  \begin{tikzcd}[row sep=-2pt]
    \Omega^p_{\mathrm{dR}}(X^{D})
    \ar[
      rr,
      "{
        \differential
      }"
    ]
    &&
    \Omega^{p+1}_{\mathrm{dR}}(X^{D})
    \\[+5pt]
    \Omega^p_{\mathrm{dR}}(X^{D})_{\iota_{\partial_t} = 0}
    \ar[
      rr,
      "{
        \differential_s
      }"
    ]
    \ar[
      ddrr,
      "{
        \differential_t
      }"{description}
    ]
    &&
    \Omega^{p+1}_{\mathrm{dR}}(X^{D})_{\iota_{\partial_t} = 0}
    \\
    \oplus
    &&
    \oplus
    \\
    \star
    \,
    \Omega^p_{\mathrm{dR}}(X^{D})_{\iota_{\partial_t} = 0}
    \ar[
      rr,
      "{
        \differential_s
      }"{swap}
    ]
    &&
    \star
    \,
    \Omega^{p-1}_{\mathrm{dR}}(X^{D})_{\iota_{\partial_t} = 0}
  \end{tikzcd}
$$
making it the total differential of a bicomplex:
\begin{equation}
  \label{BicomplexStructure}
  \differential_t \circ \differential_t
   \;=\; 0
   \,,\hspace{16pt}
  \differential_s \circ \differential_s
  \;=\;
  0
  \,,\hspace{16pt}
  \differential_t \circ \differential_s
  \;=\;
  - \differential_s \circ \differential_t
  \,.
\end{equation}
\end{itemize}

\medskip

\noindent
{\bf Higher Maxwell equations in duality-symmetric form.} By a system of {\it higher flux densities} on $X^D$ satisfying
{\it higher Maxwell equations} of a higher gauge field theory, we mean here a system of differential equations of this form:
\begin{equation}
  \label{TheCovariantEquationsOfMotion}
  \hspace{-10pt}
  \adjustbox{
    fbox
  }{$
  \hspace{13pt}
  \def\arraystretch{2}
  \begin{array}{c}
    \overset{
      \mathclap{
        \raisebox{3pt}{
        \scalebox{.7}{
          \color{gray}
          \bf
          Higher Maxwell equation
        }
        }
      }
    }{
    \differential
    \,
    \vec F
    \;=\;
    \vec P\big(
      \vec F
    \big)
    }
    \\
    \underset{
     \mathclap{
       \raisebox{-3pt}{
         \scalebox{.7}{
           \color{gray}
           \bf
           Constitutive equation
         }
       }
     }
    }{
    \star \, \vec F
    \,=\,
    \vec \mu\big(
      \vec F
    \big)
    }
  \end{array}
  \hspace{1cm}
  \def\arraystretch{1.6}
  \begin{array}{l}
    \overset{
      \mathclap{
        \raisebox{3pt}{
          \scalebox{.7}{
            \color{gray}
            \bf
            \def\arraystretch{.9}
            \begin{tabular}{c}
              index set of
              \\
              flux species
            \end{tabular}
          }
        }
      }
    }{
      I \,\in\, \mathrm{Sets}
    },
    \;
    \overset{
      \mathclap{
        \raisebox{2pt}{
          \scalebox{.7}{
            \color{gray}
            \bf
            degrees
          }
        }
      }
    }{
      \big(\mathrm{deg}_i \,\in\,
      \mathbb{N}_{\geq 1}
      \big)_{i \in I}
    }
    \,,\;\;
    \overset{
      \mathclap{
       \raisebox{2pt}{
         \scalebox{.7}{
           \color{gray}
           \bf
           flux densities
         }
       }
      }
    }{
    \vec F
    \,\defneq\,
    \Big(
    F^{(i)}
    \;\in\;
  \Omega^{\mathrm{deg}_i}_{\mathrm{dR}}
  \big(X^D\big)
  \Big)_{i \in I}
  }
  \\
  \underset{
    \mathclap{
      \hspace{-3cm}
      \raisebox{2pt}{
      \scalebox{.7}{
        \color{gray}
        \bf
        fluxes sourcing fluxes
        }
      }
    }
  }{
  \vec P
  \,\defneq\,
  \big(
    \mbox{$P^{(i)}$ grd. symm. polynomial}
  \big)_{i \in I}
  }
  \,,\;\;
  \underset{
    \mathclap{
      \raisebox{-2pt}{
        \scalebox{.7}{
          \color{gray}
          \bf vacuum permittivity
        }
      }
    }
  }{
  \mbox{
    $\vec \mu$
    invertbl. matrix
  }
  }
  \end{array}
  \hspace{-6pt}
  $}
\end{equation}

\smallskip
\noindent whose space of solutions we refer to as the {\it solution space of flux densities}: \footnote{
In \cref{FluxQuantizationAndPhaseSpace} we regard the solution space \eqref{SolutionSpace} as a 0-truncated smooth stack (a {\it smooth set}), but for
the moment the reader may think of it as just a plain set.
}
\begin{equation}
  \label{SolutionSpace}
  \mathrm{SolSpace}
  \,\defneq\,
  \left\{
    \vec F
    \,\defneq\,
    \Big(
    F^{(i)}
    \,\in\,
    \Omega^{\mathrm{deg}_i}_{\mathrm{dR}}\big(
      X^D
    \big)
  \Big)
  \,\middle\vert\,
  \def\arraystretch{1.4}
  \begin{array}{l}
    \differential
    \,
    \vec F \,=\,
    \vec P\big(\vec F\big)
    \\
    \star \, \vec F
    \,=\,
    \vec \mu\big(\vec F \big)
  \end{array}
 \!\! \right\}
  \!.
\end{equation}

\smallskip

\noindent Notice that we do not assume the higher Maxwell equations to be Euler-Lagrange equations of a Lagrangian density,
nor do Lagrangian densities play any role here; cf. Rem. \ref{RelationToTraditionalConstructionOfPhaseSpace} below.

\medskip

\noindent
{\bf Canonical formulation.}
Under the space/time-decomposition \eqref{DecompositionOfDifferentialForms}, the covariant equations of motion
\eqref{TheCovariantEquationsOfMotion} on the uniquely decomposed flux densities
\begin{equation}
  \label{CanonicalDecompositionOfFluxDensities}
  \hspace{-1cm}
  \vec F
  \;=\;
  \vec B
  \,+\,
  \star \,
  \vec E
  \,,
  \hspace{1cm}
  \mbox{for}
  \left\{\!\!
  \def\arraystretch{1.9}
  \begin{array}{l}
    \overset{
      \mathrlap{
        \hspace{35pt}
        \raisebox{3pt}{
          \scalebox{.7}{
            \color{gray}
            \bf
            Magnetic flux densities
          }
        }
      }
    }{
    \vec B
    \,\defneq\,
    }
    \Big(
      B^{(i)}
      \,\in\,
      \Omega^{\mathrm{deg}_i}_{\mathrm{dR}}\big(
        X^D
      \big)_{\iota_{\partial_t} = 0}
    \Big)_{i \in I}
    \\
    \underset{
      \mathrlap{
        \hspace{34pt}
        \raisebox{-7pt}{
          \scalebox{.7}{
            \color{gray}
            \bf
            Electric flux densities
          }
        }
      }
    }{
    \vec E
    \,\defneq\,
    }
    \Big(
      E^{(i)}
      \,\in\,
      \Omega^{D-\mathrm{deg}_i}_{\mathrm{dR}}\big(
        X^D
      \big)_{\iota_{\partial_t} = 0}
    \Big)_{i \in I}
  \end{array}
  \right.
\end{equation}
become:
\begin{equation}
  \label{GaussFaradayAmpereEquations}
  \hspace{-.1cm}
  \def\arraystretch{1.6}
  \def\arraycolsep{20pt}
  \begin{array}{rcl}
    \differential
    \,
    \vec F
    \;=\;
    \vec P\big(
      \vec F \,
    \big)
    &\Leftrightarrow&
    \left\{\!
    \def\arraystretch{1.6}
    \def\arraycolsep{7pt}
    \begin{array}{l}
      \overset{
        \mathclap{
          \raisebox{3pt}{
            \scalebox{.7}{
              \color{gray}
              \bf
              Higher Gau{\ss} law
            }
          }
        }
      }{
      \differential_s
      \,
      \vec B
      \,=\,
      \vec P\big(
        \vec B
      \big)
      }
      \\
      \underset{
        \mathclap{
          \hspace{-2.8cm}
          \raisebox{-3pt}{
            \scalebox{.7}{
              \color{gray}
              \bf
              Higher Faraday-Amp{\`e}re law
            }
          }
        }
      }{
      \differential_t
      \,
      \vec B
      \,=\,
      -
      \differential_s
      \star
      \vec E
      \,+\,
      (\star E^{(i)})
      \wedge
      \frac{\delta}{\delta B^{(i)}}
      \vec P\big(
        \vec B
      \big)
      }
    \end{array}
    \right.
    \\
    \Downarrow
    \hspace{30pt}
    \;
    \\
    \differential
    \,
    \vec P\big(
      \vec F
    \big)
    \,=\,
    0
    &\Leftrightarrow&
    \left\{\!\!\!\!\!\!\!\!
    \def\arraystretch{1.6}
    \begin{array}{l}
      \differential_s
      \,
      \vec P
      \big(
        \vec B
      \big)
      \,=\,
      0
      \\
      \differential_t
      \,
      \vec P\big(
        \vec B
      \big)
      \underset{
        \mathclap{
          \hspace{-5pt}
          \raisebox{-8pt}{
            \scalebox{.7}{
              \color{gray}
              \bf
              Integrability condition
            }
          }
        }
      }{
        \,=\,
      }
      -
      \differential_s
      \Big(
      \big(
        \star E^{(i)}
      \big)
      \wedge
      \frac{\delta}
      {\delta B^{(i)}}
      \vec P
      \big(
        \vec B
      \big)
      \Big)
    \end{array}
    \right.
    \\[30pt]
    \star \, \vec F
    \,=\,
    \vec \mu\big(\vec F\big)
    &\Leftrightarrow&
    \left\{\!\!\!\!\!\!\!
    \begin{array}{l}
    \vec E
    \underset{
      \mathrlap{
        \hspace{8pt}
        \raisebox{-5pt}{
          \scalebox{.7}{
            \color{gray}
            \bf
            Duality relation
          }
        }
      }
    }{
      \,=\,
    }
    (-1)^{(D-\mathrm{deg}_i)\mathrm{deg}_i}
    \,
    \vec \mu \big(\vec B\big)
    \mathrlap{\,.}
    \end{array}
    \right.
  \end{array}
\end{equation}

\begin{lemma}[Gau{\ss} law is first class constraint]
\label{GaussLawIsFirstClassConstraint}
The Gau{\ss} law \eqref{GaussFaradayAmpereEquations} is preserved by time evolution: If the Faraday-Amp{\`e}re law holds on
$X^D$ and the Gau{\ss} laws holds on a Cauchy surface $X^d$, then the Gau{\ss} law also holds on all of $X^D$.
\end{lemma}
\begin{proof}
It is sufficient to observe that the time derivative of the Gau{\ss} law vanishes:
\[
    \differential_t
    \,
    \Big(
      \differential_s
      \,
      \vec B
      \,-\,
      \vec P\big(
        \vec B
      \big)
    \Big)
    \;=\;
    -
    \differential_s
    \,
    \differential_t
    \,
    \vec B
    -
    \differential_t
    \,
    \vec P\big(
      \vec B
    \big)
    \;=\;
    0
    \,.
\]
Here the last step is immediate by using the Faraday-Amp{\`e}re law on the first summand and its integrability condition on the second.
(We have recorded more details in \cref{Computations}.)
\end{proof}
\begin{theorem}[Canonical solution space]
  \label{PhaseSpaceViaGaussLaws}
  The solution space to \eqref{TheCovariantEquationsOfMotion} is identified with the space of $\vec B$-fields on a Cauchy surface satisfying (just) the Gau{\ss} law:
  \begin{equation}
    \label{CanonicalSolutionSpace}
    \mathrm{SolSpace}
    \;\;\simeq\;\;
    \left\{
      \vec B
      \,\defneq\,
    \Big(
      B^{(i)} \,\in\,
      \Omega^{\mathrm{deg}_i}_{\mathrm{dR}}
      (X^d)
    \Big)_{i \in I}
    \,\middle\vert\,
    \begin{array}{l}
     \differential
     \,
     \vec B
     \,=\,
     \vec P
     \big(
       \vec B
     \big)
    \end{array}
    \right\}
    \,.
  \end{equation}
\end{theorem}
\begin{proof}
  By the duality relation, the $\vec E$-fields are just different names for $\vec B$-fields.
  Inserting the duality relation into the Faraday-Amp{\`e}re law makes the latter a first-order differential equation on the $\vec B$-fields,
  with a unique solution on $X^D$ for every choice of $\vec B$-fields on $X^d$. For every such solution, Lem. \ref{GaussLawIsFirstClassConstraint}
  says that it is also a solution of the Gau{\ss} law iff it is so on a Cauchy surface.
\end{proof}

\newpage
\begin{remark}[Cohomological nature of the canonical solution space]
\label{CanonicalNatureOfSolutionSpace}
The import of Thm. \ref{PhaseSpaceViaGaussLaws} is to exhibit the canonical solution space as a purely cohomological
construction -- cf. \eqref{SolutionSpaceViaLInfinity} below -- which, {\it at face value}, is independent of the
pseudo-Riemannian structure encoded in the Hodge-star operator. The latter is all absorbed into the isomorphism \eqref{CanonicalSolutionSpace} between
cohomologically constrained initial value data on a Cauchy surface
and
the corresponding solutions on all of spacetime. This isomorphism is given by the Faraday-Amp{\'e}re law \eqref{GaussFaradayAmpereEquations}
which maps points in canonical phase space to actual field histories on spacetime, using the duality relation to identify the $\vec E$-fields with $\vec B$-fields.

This is noteworthy but not surprising. Already in the familiar case of the phase space of vacuum electromagnetism (recalled in \cref{MaxwellField}),
it is well-known that the electric flux density is entirely independent of the magnetic flux density on a Cauchy surface, the former being
the differential $\differential A$ of the canonical coordinate $A$ and the latter being the corresponding canonical momentum $E$.
\end{remark}

\begin{remark}[\bf The moduli problem of solutions]
  \label{ModuliProblemOfSolutions}
  For all theories of interest, the functor of solution sets \eqref{PhaseSpaceViaGaussLaws}
  of their higher Maxwell equations  is representable
  \begin{equation}
    \label{Representability}
    \Big\{
      \vec B \,\in\,
      \Omega^{\vec {\mathrm{deg}}}_{\mathrm{dR}}(-)
      \,\big\vert\,
      \differential \, B
      \,=\,
      \vec P\big(
        \vec B
      \big)
    \Big\}
    \;\simeq\;
    \mathrm{Hom}_{\mathrm{dgAlg}}\big(
      \mathrm{CE}(\mathfrak{a})
      ,\,
      \Omega^\bullet_{\mathrm{dR}}(-)
    \big)
  \end{equation}

 \noindent
 by a dgc-algebra (cf. \cite[\S 4]{Char}) $\mathrm{CE}(\mathfrak{a})$,
  which is the quotient of the free dgc-algebra on the flux species by the abstract Gau{\ss} law relation (cf. \cite[Ex. 4.15]{Char}):
  \begin{equation}
    \label{CEOfCharacteristicLInfinityAlgebra}
    \mathrm{CE}(\mathfrak{a})
    \;\defneq\;
    \mathbb{R}
    \Big[
      \vec b
      \,\defneq\,
      \big(
        b^{(i)}_{\mathrm{deg}_i}
      \big)_{i \in I}
    \Big]
    \Big/ \!\!
    \left(
      \differential
      \,
      \vec b \,=\, \vec P\big(
        \vec b \,
      \big)
    \! \right)
    \!.
  \end{equation}
  This requires that the integrability condition holds purely algebraically, namely when differential forms $\vec B$ are
replaced by abstract algebra generators $\vec b$ of the same degree:
  $
    \differential
    \, \vec P\big( \vec b \big)
    \,=\,
    0
    \,.
  $ If there are finitely many $b^{(i)}_{\mathrm{deg}_i}$ for each degree $\mathrm{deg}_i$ then this means equivalently that $\mathrm{CE}(\mathfrak{a})$
  is the Chevalley-Eilenberg algebra of an $L_\infty$-algebra $\mathfrak{a}$ (cf. \cite[Def. 4.13]{Char}), the {\it characteristic $L_\infty$-algebra}
  of the given higher gauge theory.

  \smallskip

  In this case the right hand side of \eqref{Representability} is a functor of closed $\mathfrak{a}$-valued differential forms (the ``Maurer-Cartan elements'' in $\Omega^\bullet_{\mathrm{dR}}(-) \otimes \mathfrak{g}$) \cite[Def. 6.1]{Char}:
  \begin{equation}
  \label{SolutionSpaceViaLInfinity}
  \mathrm{SolSpace}
  \;\simeq\;
  \Omega_{\mathrm{dR}}\big(
    X^{d}
    ;\,
    \mathfrak{a}
  \big)_{\closed}
\,.
\end{equation}
\end{remark}
This observation is the key to understanding flux quantization laws via non-abelian differential cohomology \cite{Char}, to which we turn in \cref{FluxQuantizationAndPhaseSpace}

\medskip

\subsection{Flux quantization and Phase Space}
\label{FluxQuantizationAndPhaseSpace}

We explain how, in view of Thm. \ref{PhaseSpaceViaGaussLaws} and Rem. \ref{CanonicalNatureOfSolutionSpace}, the general principle behind flux quantization
on phase space is now almost self-evident: passage to the non-abelian differential cohomology of \cite{Char}.

\medskip

\noindent
{\bf Smooth sets and Smooth groupoids.}
One wants to understand the solution space \eqref{CanonicalSolutionSpace} not just as a bare set, but as a differential-geometric space.
If $X^d$ is compact then one can equip it with the structure of an infinite-dimensional Fr{\'e}chet manifold, but in general not even that
is possible. However, in the incarnation \eqref{SolutionSpaceViaLInfinity} it canonically carries the convenient structure of a
``smooth set'' \cite{GiotopoulosSati23}\cite{Schreiber24}, which just means to regard the whole system of solutions in arbitrary {\it smooth families}
indexed by $p$-parameter spaces $\mathbb{R}^p$ for all $p$:
\medskip
\begin{equation}
  \label{SolutionSpaceAsSmoothSet}
  \mathbf{\Omega}_{\mathrm{dR}}\big(
    X^d
    ;\, \mathfrak{a}
  \big)
  \;:\;
  \mathbb{R}^p
  \;\mapsto\;
  \Omega_{\mathrm{dR}}\big(
    X^d \times \mathbb{R}^p
    ;\,
    \mathfrak{a}
  \big)_\closed
  \,.
\end{equation}
(Notice that solutions in smooth families do make good sense due to Thm. \ref{PhaseSpaceViaGaussLaws}, as discussed in
Rem. \ref{CanonicalNatureOfSolutionSpace}.)

\smallskip

Technically this means to regard $\mathrm{SolSpace}$ as a ``0-truncated object'' in the  ``$\infty$-topos of smooth $\infty$-groupoids'' (cf. \cite[pp. 41]{Char}\cite[\S 1]{AlfonsiYoung23}). While we will invoke some of this language of stacky differential homotopy theory now, in order to be precise, the reader unfamiliar with this technology may, without much loss, think of the following Def. \ref{FluxQuantizationLaws} as a formal machine which adjoins to the solution space of flux densities all the possible corresponding {\it gauge potentials}, while also taking care of gauge redundancy (cf. Rem. \ref{StackyPhaseSpaceReduction}).

\medskip

\noindent
{\bf Flux quantization.} In the language of stacky homotopy theory, we naturally arrive at the following notions (for more survey see \cite{SS24Flux}):

\begin{definition}[Flux quantization law]
\label{FluxQuantizationLaws}
Given higher flux densities \eqref{TheCovariantEquationsOfMotion} with characteristic $L_\infty$-algebra $\mathfrak{a}$ \eqref{CEOfCharacteristicLInfinityAlgebra}, a {\it flux quantization law} is a choice of (homotopy type of) a classifying space\footnote{
  For our purposes here a ``classifying space'' is a simply connected topological space with finite-dimensional rational cohomology in each degree, so that the fundamental theorem of dg-algebra rational homotopy theory applies to it, as reviewed in \cite[\S 5]{Char}. One can relax these assumptions (cf. \cite[Rem. 5.1]{Char}), and will want to do so, but this goes beyond the intended scope of this note.
} $\mathcal{A}$ whose rational Whitehead $L_\infty$-algebra \cite[Prop. 5.11]{Char} is $\mathfrak{l}A \,\simeq\, \mathfrak{a}$.
\end{definition}
\begin{definition}[Flux-quantized phase space]
\label{IntrinsicPhaseSpace}
Given a flux quantization law $\mathcal{A}$ (Def. \ref{FluxQuantizationLaws}), we say that the {\it Phase space} of the higher gauge theory
is the moduli stack of differential $\mathcal{A}$-cohomology \cite[Def. 9.3]{Char} on $X^d$:
\begin{equation}
  \label{PhaseSpace}
  \mathrm{PhaseSpace}
  \;\defneq\;
    \mathbf{\Omega}_{\mathrm{dR}}\big(
      X^d
      ;\,
      \mathfrak{l}\mathcal{A}
    \big)_{\closed}
    \underset{
      {
        L^{\mathbb{R}}
        \mathcal{A}(X^d)
      }
    }{\times}
    \mathcal{A}\big(X^d\big)
    \,,
\end{equation}
where the notation on the right denotes the mapping stack from $X^d$ into the homotopy pullback of the $\mathbb{R}$-rationalization \cite[Def. 5.7]{Char}
of the homotopy type $\mathcal{A}$, using its (non-abelian) differential character map \cite[Def. 9.2]{Char}.
\end{definition}
\begin{remark}[Scope of examples]
\label{ScopeOfExamples}
$\,$

 \noindent {\bf (i)}  For $\mathcal{A} \underset{\mathrm{whe}}{\,\simeq\,} B^n \mathrm{U}(1) \underset{\mathrm{whe}}{\,\simeq\,}
 B^{n+1}\mathbb{Z} \,\defneq\, K(\mathbb{Z}, n+1)$,  the right hand side of \eqref{PhaseSpace} reduces \cite[Prop. 9.5]{Char} to ordinary
 integral differential cohomology (Cheeger-Simons characters, cf. \cite[\S 9]{Char}) modeled for $n =1$ by principal $\mathrm{U}(1)$-bundles
 with connection  and for $n = 2$ by $\mathrm{U}(1)$-bundle gerbe connection, etc.; see the examples in \cref{MaxwellField} \& \cref{ChiralBoson} below.

 \noindent {\bf (ii)}   More generally, for $\mathcal{A} \,=\, E_n$ the $n$th component in a {\it spectrum} ($\Omega$-spectrum) of pointed topological spaces,
 the right-hand side of \eqref{PhaseSpace} reduces \cite[Ex. 9.1]{Char} to the corresponding Whitehead-generalized ``canonical'' differential cohomology
 theory $\widehat E(-)$. (For $E = \mathrm{KU}$ this is the case of differential K-theory, cf. \cite[Ex. 9.2]{Char} and \cref{RRFields} below).
 This is the generality in which charge quantization by differential cohomology has been discussed in most of the literature.

 \noindent {\bf (iii)}   But as soon as the higher gauge theory has non-linearities where combinations of fluxes act as sources for other fluxes
 (as is the case for the 11d SuGra C-field, cf. \cref{TheCField}), no Whitehead-generalized cohomology theory is an admissible flux quantization
 in the sense of Def. \ref{FluxQuantizationLaws}. Instead, in these cases the flux-quantization law is given by a classifying space with
 non-trivial Whitehead $L_\infty$-brackets and and by a non-abelian differential cohomology theory in the sense of \cite{Char}.
\end{remark}

\begin{remark}[Relation to traditional construction of phase space]
  \label{RelationToTraditionalConstructionOfPhaseSpace}
  $\,$

 \noindent {\bf (i)}   The traditional way to arrive at the phase space of a higher gauge theory is completely different
 from Def. \ref{IntrinsicPhaseSpace}. In particular, our construction does not refer to a Lagrangian density.

   \noindent {\bf (ii)}  That Def. \ref{IntrinsicPhaseSpace} is nevertheless the correct definition of phase spaces of higher gauge
   theory follows from Thm. \ref{PhaseSpaceViaGaussLaws} and Def. \ref{IntrinsicPhaseSpace}, saying that choices of gauge potentials for given
   on-shell fluxes are lifts of the latter to cocycles in differential cohomology \eqref{PhaseSpaceOverSolutionSpace}.

 \noindent {\bf (iii)}   Concretely, we see below that Def. \ref{IntrinsicPhaseSpace} reproduces the traditional phase space of vacuum
 electromagnetism \eqref{TraditionalEMPhaseSpace}.

 \noindent {\bf (iv)}   What we are not discussing here is the symplectic structure on phase space, hence the Poisson brackets of
 gauge fields and eventually their quantization, this is instead the topic of the companion article \cite{SS23QuantumFluxObservables}.

\end{remark}

\begin{remark}[Stacky phase space reduction and integrated BRST]
\label{StackyPhaseSpaceReduction}
 $\,$

\noindent {\bf (i)}  The construction of the phase space in Def. \ref{FluxQuantizationLaws} as a {\it mapping stack} immediately makes it the
  correct ``reduced'' phase space of the gauge theory, in the refined stacky sense:

 \noindent {\bf (ii)}  Namely, the Hamiltonian evolutions induced by the (higher) Gau{\ss}-law constraints \eqref{GaussFaradayAmpereEquations} are (e.g. \cite[\S 1.2]{HenneauxTeitelboim92})
  the {\it gauge transformations}, translating among (higher) gauge potentials that correspond to the same flux densities. In traditional approaches,
  one is tempted to quotient out these gauge equivalences to arrive at the naive {\it reduced phase space}, but refrains from doing so due to
  the bad technical behavior of the ordinary quotient space \cite[\S 2.2.3]{HenneauxTeitelboim92}.

 \noindent {\bf (iii)}  Indeed, more properly, gauge transformations should be
  retained as isomorphisms in a smooth groupoid structure on phase space, which is instead the correct stacky {\it homotopy quotient}
  (cf. \cite[Ntn. 3.1.41, Prop. 3.2.76]{SatiSchreiber21}) by the gauge equivalences, locally. This (higher) groupoid structure may be understood (cf. \cite[\S 5.3]{AlfonsiYoung23})
  as
  the Lie integration (hence the global non-perturbative version) of the corresponding BRST-complex, taking infinitsimal ``ghost'' fields to finite (and large) gauge transformations.

 \noindent {\bf (iv)}  All of this is automatically embodied by the construction in Def. \ref{FluxQuantizationLaws}.
  For example, for $\mathcal{A} \,=\, B \mathrm{U}(1)$ the flux quantized phase space is the stack $\mathrm{U}(1)\mathrm{Bund}(X^d)_{\mathrm{conn}}$
  of principal circle-bundles with connection, which locally on a chart $U^d \hookrightarrow X^d$ is the homotopy quotient of local gauge potentials
  $\mathbf{\Omega}^1_{\mathrm{dR}}(U^d)$ by the smooth group of gauge transformations $\mathrm{Map}(U^d, \mathrm{U}(1))$.
\end{remark}

\newpage

\section{Examples and Applications}
\label{ExamplesAndApplications}

Two basic examples

\cref{MaxwellField} -- The vacuum Maxwell field

\cref{ChiralBoson} -- The chiral boson

\noindent
serve to connect the generalized formulas from \cref{FluxQuantizationOnPhaseSpace} to traditional expressions and proposals;
but even here the resulting space of choices of flux quantization laws on phase space has not received attention before (cf. \cite{SS23QuantumFluxObservables}).

\smallskip

The example

\cref{RRFields} -- The 10d SuGra RR-fields

\noindent
seems noteworthy in that the discussion in \cref{FluxQuantizationOnPhaseSpace} recovers the famous {\it Hypothesis K} of charge quantization
in (twisted) topological K-theory, but, by applying this to fields on a Cauchy surface, without the otherwise notoriously elusive duality constraint.

\smallskip

Finally, the example

\cref{TheCField} -- The 11d SuGra C-field

\noindent
is a genuinely non-abelian case, out of reach of previous flux quantization in abelian Whitehead-generalized cohomology theory.
Our main point here is to observe, again, that the {\it Hypothesis H} on its flux quantization in Cohomotopy already produces
the full phase space, without further need of a duality constraint.

\smallskip

\subsection{The vacuum Maxwell field}
\label{MaxwellField}

For a recollection of Maxwell's equations expressed in differential forms see \cite[\S 3.5 \& 7.2b]{Frankel97}: The traditional electromagnetic field ``3-vectors'' $\big(E^1, E^2, E^3\big)$ and $(B^1, B^2, B^3)$ with respect to any coordinate chart $\mathbb{R}^3 \xhookrightarrow{\iota} X^3$ of a spatial Cauchy surface dualize to electromagnetic flux densities
$$
  \iota^\ast B
    \;\defneq\;
  \tfrac{1}{2}
  \,
  B^i
  \epsilon_{i j k}
  \,
  \mathrm{d}x^j
  \wedge
  \mathrm{d}x^k
  \,,\;\;\;\;\;
  \iota^\ast E
    \;\defneq\;
  \tfrac{1}{2}
  \,
  E^i
  \epsilon_{i j k}
  \,
  \mathrm{d}x^j
  \wedge
  \mathrm{d}x^k
  \,,
$$
which in turn
combine into to the Faraday tensor $F$ according to \eqref{CanonicalDecompositionOfFluxDensities}:
\[
  \def\arraystretch{1.3}
  \begin{array}{r}
    F
    \,=\,
    B \,-\, \star E
    \\
    \star F
    \,=\,
    E \,+\, \star B
  \end{array}
  \hspace{1cm}
  \mbox{for}
  \begin{array}{l}
    B, E \,\in\, \Omega^2_{\mathrm{dR}}(X^4)_{\iota_{\partial_t} = 0}
    \,.
  \end{array}
\]

The vacuum Maxwell equations on $X^4$ in the ``premetric'' form \eqref{TheCovariantEquationsOfMotion} are (\cite[\S 80]{Cartan24}, cf. \cite[Ex. 3.8]{Freed00}\cite{HehlItinObukhov16}\cite[Rem. 2.3]{BBSS17}\cite[Def. 1.16]{LazaroiuShahbazi22}\cite[(3)]{LazaroiuShahbazi23}):
\begin{equation}
  \label{VacuumMaxwellEquations}
  \begin{array}{l}
    \differential
    \,
    F
    \,=\, 0
    \\
    \differential
    \,
    G
    \,=\, 0
    \\
    G \,=\, \star \, F
  \end{array}
  \hspace{1cm}\mbox{for}
  \begin{array}{l}
    F, \, G
    \,\in\,
    \Omega^2_{\mathrm{dR}}(X^4)
    \,.
  \end{array}
\end{equation}
The solution space \eqref{CanonicalSolutionSpace} to the ordinary Maxwell's equations is
\begin{equation}
  \label{SolutionSpaceOfMaxwellTheory}
  \mathrm{SolSpace}_{\scalebox{.7}{A-Field}}
  \;\;
  =
  \;\;
  \Big\{
    B, E\,
    \in
    \,
    \Omega^2_{\mathrm{dR}}(X^3)
    \,\big\vert\,
    \differential \, B \,=\, 0
    \,,
    \differential \, E \,=\, 0
  \Big\}
  \,,
\end{equation}
where on the right we have the historical magnetic and electric Gau{\ss} laws. The characteristic $L_\infty$-algebra
\eqref{CEOfCharacteristicLInfinityAlgebra} is hence the direct sum of {\it two} copies of the shifted line algebra (cf. \cite[Ex. 4.12]{Char}):
\[
  \mathrm{SolSpace}
  \;\simeq\;
  \Omega^2_{\mathrm{dR}}(X^3)_\closed
  \oplus
  \Omega^2_{\mathrm{dR}}(X^3)_\closed
  \;\simeq\;
  \Omega_{\mathrm{dR}}\big(
    X^3
    ;\,
    b\mathfrak{u}(1)
    \oplus
    b\mathfrak{u}(1)
  \big)_\closed
  \,,
\]
which means that the possible flux quantization laws (Def. \ref{FluxQuantizationLaws}) are given by homotopy types $\mathcal{A}$ whose rationalization is
\medskip
$$
  L^{\mathbb{Q}} \mathcal{A}
  \;\simeq\;
  B^2 \mathbb{Q}
  \times
  B^2 \mathbb{Q}
  \,.
$$

\smallskip

Now, traditional Dirac flux quantization (\cite{Dirac31}\cite[\S 2]{Alvarez85}\cite[Ex. 2.12]{Freed00}) corresponds essentially to the choice
and $\mathcal{A} \,\defneq\, B \mathrm{U}(1) \times B^2 \mathbb{Q}$
and
promotes the $B$-flux density to a cocycle in integral differential cohomology (cf. Rem. \ref{ScopeOfExamples}) while the $E$-flux density remains
essentially unconstrained, whence the flux-quantized phase space \eqref{IntrinsicPhaseSpace} in this case becomes
\begin{equation}
  \label{TraditionalEMPhaseSpace}
  \mathrm{PhaseSpace}_{
    \scalebox{.7}{A-Field}
  }^{\mathrm{Dirac}}
  \;\defneq\;
  \left\{\!\!
    \def\arraystretch{1.6}
    \begin{array}{l}
    \widehat{A}
    \,\in\,
    \Omega^2_{\mathrm{dR}}(X^3)_\closed
    \times_{B^2 \mathbb{R}}
    B \mathrm{U}(1)
    \\
    E\,
    \in
    \,
    \Omega^2_{\mathrm{dR}}(X^3)
    \end{array}
    \,\middle\vert\,
    \def\arraystretch{1.3}
    \begin{array}{l}
      \differential \, E \,=\, 0
    \end{array}
  \right\}
  \,.
\end{equation}
Here the expression in the first line is (by \cite[Prop. 9.5]{Char}, cf. e.g. \cite[2.5]{FSS13}\cite[\S 2.5]{FSS14Stacky}) equivalently the groupoid of $\mathrm{U}(1)$-principal bundles $P$ with connection $\nabla$ over $X^3$, whose morphisms are gauge transformations $g$:
\begin{equation}
  \label{GroupoidOFCircleConnectionsManifest}
  \Omega^2_{\mathrm{dR}}\big(
    X^3
  \big)_{\mathrm{clsd}}
  \underset{
    B^2 \mathbb{R}
  }{\times}
  B \mathrm{U}(1)
  \;\;
    \simeq
  \;\;
  \mathrm{U}(1)\mathbf{Bund}_{\mathrm{conn}}(X^3)
  \;\;
  =
  \;\;
  \left\{
  \begin{tikzcd}[
    column sep=10pt
  ]
    (P,\nabla)
    \ar[
      rr,
      bend left=20pt,
      "{
        g
      }",
      "{\ }"{name=s, swap}
    ]
    \ar[
      rr,
      bend right=20pt,
      "{
        g'
      }"{swap},
      "{\ }"{name=t}
    ]
    &&
    (P', \nabla')
  \end{tikzcd}
  \right\}
  \,,
\end{equation}
which is well-known to model the flux-quantized electromagnetic field (\cite{WuYang75}, review includes \cite[\S 5.5]{EGH80}\cite{RudolphSchmidt17}). As such, \eqref{TraditionalEMPhaseSpace}  coincides with the standard canonical phase space of vacuum electromagnetism in temporal gauge (cf. \cite[\S III]{Corichi98}\cite[\S 5]{BlaschkeGieres21}):
\begin{itemize}[leftmargin=.4cm]
\item The ``canonical coordinate'' is $\widehat{A} \,=\, (P,\nabla)$, typically expressed with respect to a local trivialization of $P$ over surjective submetion $Y \twoheadrightarrow X$ (such as $Y = P$) as a ``vector potential'' 1-form $A \,\in\, \Omega^1_{\mathrm{dR}}(Y)$ satisfying further conditions;
\item the
electric flux density $E$ is\footnote{We follow \cite{CattaneoPerez17}\cite[\S A.1]{SS23QuantumFluxObservables} in regarding the electric flux density $E$ as a 2-form, as befits a flux density, whereas
many authors regard it as a ``3-vector'', the Hodge dual 1-form $\star E$. It is due to this implicit dualization that the divergence operation in the traditional
Gau{\ss} law becomes the closure operation $\differential E = 0$.}
the canonical momentum to the canonical coordinate $\widehat{A}$ and its Gau{\ss} law is the constraint condition.
\end{itemize}

\begin{remark}[Varying monopole sectors and categorified symmetries]
  Traditional discussions typically focus attention on a single connected component of the groupoid \eqref{GroupoidOFCircleConnectionsManifest} of all $\mathrm{U}(1)$-connections, hence fix the gauge equivalence class of one background connection, or at least fix the class of the underlying $\mathrm{U}(1)$-principal bundle. This is the perspective of {\it perturbation theory} where the global topological structure of the fields is fixed and only small field perturbations about these backgrounds are considered. In contrast, the phase spaces obtained here are automatically non-perturbative in that they reflect all topological ``monopole sectors'' at once.

  This is relevant: For instance there are global symmetries of vacuum electromagnetism given by ``shifting'' (namely: tensoring) the global electromagnetic field $\widehat{A} = (P,\nabla)$ by flat bundles $\widehat{A}_{0} = (P_0, \nabla_0)$, i.e. with $F_{\nabla_0} = 0$, which in turn have ``higher global symmetries''\footnote{The 2-group of endofunctors in \eqref{HigherShiftSymmetryOfEMField} makes precise the idea of the ``higher form symmetry'' of vacuum electromagnetic considered in \cite[\S 4.1]{GKSW15}.}
  between them, given by tensoring with gauge transformations $g : \widehat{A}_0 \to \widehat{A}'_0$, forming a 2-group of 2-automorphisms (endofunctors and natural transformations, cf. \cite[\S. 8.1]{BaezLauda04}) of the groupoid \eqref{GroupoidOFCircleConnectionsManifest}, as shown in the following diagram:

  \vspace{-.6cm}
  \begin{equation}
    \label{HigherShiftSymmetryOfEMField}
    \begin{tikzcd}
      \mathrm{U}(1)
      \mathbf{Bund}_{\mathrm{conn}}(X^3)
      \ar[
        rr,
        bend left=15,
        "{
          (-)\otimes
          \widehat{A}_0
        }",
        "{\ }"{swap, name=s}
      ]
      \ar[
        rr,
        bend right=15,
        "{
          (-)\otimes
          \widehat{A}'_0
        }"{swap},
        "{\ }"{name=t}
      ]
      \ar[
        from=s,
        to=t,
        Rightarrow,
        "{
          (-)
            \otimes
          g
        }"{description, pos=.4}
      ]
      &&
      \mathrm{U}(1)
      \mathbf{Bund}_{\mathrm{conn}}(X^3)
    \end{tikzcd}
  \end{equation}
  In the traditional perspective on the EM-phase space as reflecting only a single topopological sector of the gauge field, this 2-group of higher symmetries is in general not realizable, since the group operations $(-) \otimes \widehat{A}_0$ in general shift the class of the underlying $\mathrm{U}(1)$-principal bundle, $[P \otimes P_0] \,=\, [P] + [P_0] \,\in\, H^2(X^3; \mathbb{Z})$, by a torsion element $[P_0]$.

  This example serves to amplify the relevance of the full non-perturbative phase space of a (higher) gauge theory, including all topological field sectors, which is produced by Def. \ref{IntrinsicPhaseSpace}.
\end{remark}

\smallskip

So far, all this pertains to the choice of flux quantization of the electromagnetic field that used to be the usual one essentially since \cite{Dirac31}, where the magnetic flux is quantized but the electric flux is left essentially unconstrained.
However, one may consider other flux quantization laws for Maxwell theory on phase space:
Symmetry may suggest to subject the electric flux density $E$ to the same quantization law as the magnetic flux density, hence to choose
$\mathcal{A} \,=\, B\mathrm{U}(1) \times B\mathrm{U}(1)$. Via Rem. \ref{ScopeOfExamples} (\cite[Prop. 9.5]{Char}) this is identified with the choice made
in \cite{FMS07a}\cite{FMS07b} \cite[Rem. 2.3]{BBSS17}\cite[Def. 4.1]{LazaroiuShahbazi22}\cite[Def. 4.3]{LazaroiuShahbazi23}:
\begin{equation}
  \label{SymmetricFluxQuantizationofMaxwellTheory}
  \mathrm{PhaseSpace}_{
    \scalebox{.7}{A-Field}
  }^{\mathrm{FMS}}
  \;\defneq\;
  \left\{\!\!
    \def\arraystretch{1.6}
    \begin{array}{l}
    \widehat{A}
    \,\in\,
    \Omega^2_{\mathrm{dR}}(X^3)_\closed
    \times_{B^2 \mathbb{R}}
    B \mathrm{U}(1)
    \\
    \widehat{A}'
    \,
    \in
    \,
    \Omega^2_{\mathrm{dR}}(X^3)_\closed
    \times_{B^2 \mathbb{R}}
    B \mathrm{U}(1)
    \end{array}
  \!\! \right\}
  \,.
\end{equation}
(The focus of \cite{FMS07a}\cite{FMS07b} is on a lift of the canonical Poisson bracket from \cref{TraditionalEMPhaseSpace} to \eqref{SymmetricFluxQuantizationofMaxwellTheory}. Here we disregard Poisson brackets, but see the companion discussion in \cite{SS23QuantumFluxObservables}.)

\begin{remark}[Duality covariance]
  While the flux-quantization choice \eqref{SymmetricFluxQuantizationofMaxwellTheory} is duality-sym-metric in that it treats magnetic and electric flux in the same way, it is not in itself duality {\it covariant}, so to say, in that it relies on a {\it choice} of what counts as magnetic and what counts as electric flux: More generally this choice could be made locally only, with suitable transition functions mixing electromagnetic flux densities relating different local choices, a situation that is discussed in some detail in \cite{LS18}\cite{LazaroiuShahbazi22}\cite{LazaroiuShahbazi23}, there thought of as a special case of {\it U-duality covariance} in supergravity theory.

  In terms of the abstract construction of phase spaces in \cref{FluxQuantizationAndPhaseSpace}, such local duality covariance is described by generalizing the classifying space $\mathcal{A}$ in Def. \ref{FluxQuantizationLaws} to a {\it fiber bundle} of classifying spaces (a ``local coefficient bundle''), thereby generalizing the (differential) cohomology theory that it classifies to a {\it twisted} (differential) cohomology theory. This is discussed in detail in \cite[\S 3, \S 11]{Char}, whereby our basic argument here straightforwardly generalizes to this case; but for brevity we will not further dwell on this point here.
\end{remark}

\smallskip

Even so, we wish to highlight that there are yet other possible flux quantization laws \cite[\S 2]{SS23QuantumFluxObservables}, whose torsion contributions need not actually respect the apparent electromagnetic duality. For instance,
for a finite group $K \to \mathbb{Z}_2$, we may choose
\[
  \mathrm{PhaseSpace}_{
    \scalebox{.7}{A-Field}
  }^{\mathrm{SS}}
  \;\defneq\;
  \bigg\{
  {\widehat{A}_{\mathrm{EM}}}
  \,\in\,
  \big(
    \Omega^2_{\mathrm{dR}}(X^3)_\closed
  \big)^2
  \times_{
    (B^2 \mathbb{R})^2
  }
  B\Big(
    \mathrm{U}(1)
    \rtimes
    \big(K \times \mathrm{U}(1)\big)
  \Big)
  \bigg\}
  \,.
\]
which leads to possibly non-trivial commutators between magnetic and ``large'' electric fluxes.

(This freedom of choosing ``global'' non-abelian structure even in abelian Yang-Mills theory has also been observed, from a different angle, in \cite{LazaroiuShahbazi22}.)

\smallskip

Which of these flux quantization laws really applies to the observable world is a phenomenological question that has received
almost no attention yet. (One exception is \cite{KitaevMooreWalker07} which claimed to discuss a potential tabletop
experiment that might check the law \eqref{SymmetricFluxQuantizationofMaxwellTheory}.) But analogous choices of flux quantization laws
need to be made also for any higher gauge field in hypothetical fundamental theories such as higher dimensional supergravity
(e.g., the RR-fields \cref{RRFields} or the C-field \cref{TheCField}), where they have more discernible effects, at least in theory.

\medskip

\subsection{Chiral boson in 2d}
\label{ChiralBoson}

This example is again elementary but instructive. In usual Lagrangian approaches, these self-dual fields are a source of seemingly
endless complications (see the discussion and pointers in \cite{Sen20}) and yet thought to be of profound relevance (cf., e.g., \cite{Witten97}).

\medskip

The higher Maxwell equations \eqref{TheCovariantEquationsOfMotion} for the ``chiral boson'' \cite{GirottiGomesRivelles92} on
$X^2 = \mathbb{R}^{0,1} \times S^1$ (in this case really: ``lower Maxwell equations'' for a ``right-moving scalar'' in 2d)  are :
\begin{equation}
  \label{EOMsForChiralBosonIn2d}
  \def\arraystretch{1.4}
  \begin{array}{l}
    \differential \, F \,=\, 0
    \\
    \star \, F \,=\, F
  \end{array}
  \hspace{1cm}
  \mbox{for}
  \;
  F \,\in\,
  \Omega^1_{\mathrm{dR}}(X^2)
  \,.
\end{equation}
The canonical decomposition \eqref{CanonicalDecompositionOfFluxDensities} of the flux density 1-form is
\[
  \def\arraystretch{1.5}
  \def\arraycolsep{10pt}
  \begin{array}{cll}
  &
  F
  \,=\,
  B + \star \, E
  \,,
  \hspace{1cm}
  &
  E, \, B
  \,\in\,
  \Omega^1_{\mathrm{dR}}(X^2)_{\iota_{\partial_t} = 0}
  \\
  \Leftrightarrow
  &
  F
  \,=\,
  b \, \differential x
  +
  e \, \differential t
  \,,
  &
  e, \, b
  \,\in\,
  \Omega^0_{\mathrm{dR}}(X^2)
  \,.
  \end{array}
\]
For the record, we note that the Faraday-Amp{\`e}re law \eqref{GaussFaradayAmpereEquations}
is
\[
  \differential_t
  \,
  B
  \,=\,
  -
  \differential_x \, E
  \;\;\;
  \Leftrightarrow
  \;\;\;
  \partial_t \, b
  \,=\,
  \partial_x \, e
  \;\;\;
  \underset{
    \mathclap{
      \scalebox{.7}{s.d.}
    }
  }{\Leftrightarrow}
  \;\;\;
  (\partial_t - \partial_x) b
  \,=\,
  0\,,
\]
(where in the last step we inserted the constitutive equation $\star F = F \;\Leftrightarrow\; e = b$) exhibiting the field as purely a ``right-mover'' (whence: ``chiral'').

However, the key point for us is that
the solution space \eqref{CanonicalSolutionSpace} is
\begin{equation}
  \label{SolutionSpaceOfChiralBoson}
  \mathrm{SolSpace}_{\scalebox{.65}{$\rchi$Bos}}
  \;\simeq\;
  \Big\{
    B
    \,=\,
    b \,\differential x
    \,\in\,
    \Omega^1_{\mathrm{dR}}(S^1)
  \Big\}
  \;\simeq\;
  \Omega_{\mathrm{dR}}\big(
    S^1
    ;\,
    b\mathfrak{u}(1)
  \big)_\closed
  \,,
\end{equation}
identifying (a {\it single} copy of) the line Lie algebra as the characteristic $L_\infty$-algebra \eqref{SolutionSpaceViaLInfinity}, so that the canonical choice
of flux quantization law (Def. \ref{FluxQuantizationLaws}) is $\mathcal{A} \defneq \mathrm{U}(1)$. By Rem. \ref{ScopeOfExamples}
(\cite[Prop. 9.5]{Char}) this makes the flux-quantized phase space (Def. \ref{IntrinsicPhaseSpace})
of the chiral boson be the degree=1 integral differential cohomology of the circle
(cf. \cite[pp. 13]{FMS07b}):
\[
  \mathrm{PhaseSpace}^{\mathrm{FMS}}_{\scalebox{.65}{$\rchi$Bos}}
  \;\defneq\;
  \Big\{
  A \,\in\,
  \Omega^1_{\mathrm{dR}}(S^1)
  \times_{B^2 \mathbb{R}}
  B\mathrm{U}(1)
  \Big\}
  \,.
\]
This coincides with the statement on \cite[p. 32]{FMS07b} (notice that there the focus is on the Poisson brackets on this phase space,
which is not our concern here, but cf. \cite{SS23QuantumFluxObservables}).

\smallskip

\begin{remark}
This highlights how in our duality-symmetric formulation {\it every} higher gauge theory (of Maxwell type) is regarded as ``generalized self-dual'' with the actual self-dual field theories subsumed as those whose duality-symmetric fields are {\it not doubled}, as in \eqref{SolutionSpaceOfChiralBoson} and in contrast to \eqref{SolutionSpaceOfMaxwellTheory}.
\end{remark}

\newpage

\subsection{The RR-field and Hypothesis K}
\label{RRFields}

On a spacetime of the form
$
  X^{10}
  \;=\;
  \mathbb{R}^{0,1}
  \times
  X^{9}
$,
consider NS-fields and RR-fields with their canonical decomposition \eqref{CanonicalDecompositionOfFluxDensities}:
$$
  \begin{array}{l}
  H_3
  \;=\;
  H^{\mathrm{mag}}_3
  \,+\,
  \star
  H_7^{\mathrm{el}}
  \,,
  \hspace{1cm}
  \mbox{for} \;
  \left\{
  \def\arraystretch{1.3}
  \begin{array}{l}
    H_3^{\mathrm{mag}}
    \,\in\,
    \Omega^3_{\mathrm{dR}}(X^{10})_{\iota_{\partial_t}=0}
    \\
    H^3_{\mathrm{el}}
    \,\in\,
    \Omega^7_{\mathrm{dR}}(X^{10})_{\iota_{\partial_t}=0}
  \end{array}
  \right.
  \\[20pt]
  H_7
  \;=\;
  H^{\mathrm{mag}}_7
  \,+\,
  \star
  H_3^{\mathrm{el}}
  \,,
  \hspace{1cm}
  \mbox{for}\;
  \left\{
  \def\arraystretch{1.3}
  \begin{array}{l}
    H_7^{\mathrm{mag}}
    \,\in\,
    \Omega^7_{\mathrm{dR}}(X^{10})_{\iota_{\partial_t}=0}
    \\
    H_3^{\mathrm{el}}
    \,\in\,
    \Omega^7_{\mathrm{dR}}(X^{10})_{\iota_{\partial_t}=0}
  \end{array}
  \right.
  \\[20pt]
  \big(
  F_{2k + \sigma}
  \;=\;
  B_{2k+\sigma}
  \,-\,
  \star
  \,
  B_{10-2k- \sigma}
  \big)_{1 \leq 2k+\sigma \leq 9}
  \,,
  \hspace{.3cm}
  \mbox{for} \;\;
  \Big\{
    B_{2k+\sigma}
    \,\in\,
    \Omega^{2k+\sigma}_{\mathrm{dR}}(X^{10})_{\iota_{\partial_t} = 0} \;,
  \end{array}
$$
where
\bigskip
$$
  \sigma
  \;=\;
  \left\{
  \def\arraystretch{1.2}
  \begin{array}{lcl}
    0 &\vert& \mbox{type IIA},
    \\
    1 &\vert& \mbox{type IIB}.
  \end{array}
  \right.
$$
The higher Maxwell equations for the NS/RR-fields in the ``duality symmetric'' or ``democratic''
form \eqref{TheCovariantEquationsOfMotion} are  (\cite[\S 3]{CJLP98}, cf. \cite[p. 3]{GradySati22}\cite[(6,7)]{MkrtchyanValach23}):
\smallskip
$$
  \def\arraystretch{1.3}
  \begin{array}{l}
    \differential
    \,
    H_3 \,=\, 0
    \,,
    \;\;\;
    \differential
    \,
    H_7 \,=\, 0
    \\
    \differential
    \,
    F_{2k+\sigma}
    \,=\,
    H_3
    \wedge
    F_{2k + \sigma - 2}
    \\
    \mathrlap{
      \star H_3 \,=\, H_7
    }
      \\
    \mathrlap{
      \star F_{2k+\sigma}
      \,=\
      F_{10- 2k -\sigma},
    }
  \end{array}
$$
where we understand that $F_{2k+\sigma} \defneq 0$ if the index is $< 1$ or $> 9$.

In our systematics,
the now traditional statement \cite{MinasianMoore97} that this points towards flux quantization in topological K-theory
(rather: twisted K-theory \cite{BouwknegtMathai01}\cite{MathaiSati}) comes about because \eqref{CanonicalSolutionSpace} is characterized
(in the sense of Rem. \ref{ModuliProblemOfSolutions}) by the Whitehead $L_\infty$-algebra of
the homotopy quotient of the classifying space for complex K-theory in degree $\sigma$ by the 2-group $B \mathrm{U}(1)$
\cite[Def. 4.10]{FSS17RationalSphere}\cite[Def. 4.6]{FSS18TDuality}\cite[Prop. 10.1]{Char}:
\smallskip
$$
  \def\arraystretch{1.9}
  \begin{array}{rcl}
  \mathrm{SolSpace}_{\scalebox{.65}{NS/RR-Field}}
  &\simeq&
  \left\{
   \def\arraystretch{1.3}
  \begin{array}{l}
  H^{\mathrm{mag}}_3
  \,\in\,
  \Omega^3_{\mathrm{dR}}(X^9)
  \\
  H^{\mathrm{mag}}_7
  \,\in\,
  \Omega^7_{\mathrm{dR}}(X^9)
  \\
  \big(
  B_{2k+\sigma}
  \,\in\,
  \Omega^{2k+\sigma}_{\mathrm{dR}}(X^{9})
  \big)_{
    1 \leq 2k+\sigma \leq 9
  }
  \end{array}
  \,\middle\vert\,
  \def\arraystretch{1.3}
  \begin{array}{l}
    \differential
    \,
    H_3
    \,=\,0
    ,\;\;
    \differential
    \,
    H_7
    \,=\,0
    \\
    \differential \,
    B_{2k+\sigma} \,=\,
    H_3
    \wedge
    B_{2k + \sigma - 2}
  \end{array}
  \right\}
  \\
  &\simeq&
  \Omega_{\mathrm{dR}}\Big(
    X^9
    ;\,
    \mathfrak{l}
    \big(
      (\Omega^\infty \Sigma^\sigma \mathrm{KU}) \!\sslash\! B \mathrm{U}(1)
    \big)
    \oplus
    \mathfrak{l}B^6\mathrm{U}(1)
  \Big)_\closed.
  \end{array}
$$

\medskip
\noindent Hence if we apply flux quantization for RR-fields on phase space according to Def. \ref{FluxQuantizationLaws}, then the flux quantization
law may be taken to be (twisted) complex K-theory as in the now traditional ``Hypothesis K'' (to use the term from \cite[Rem. 4.1]{SS23AnyonicDefect}
for the flux quantization hypothesis due to \cite{MooreWitten00}\cite{FreedHopkins00}\cite{Freed00}, for more see \cite{GradySati22})
but without further duality constraint, in higher but otherwise immediate analogy to the transparent case of the chiral boson (\cref{ChiralBoson}).

\smallskip
In contrast, in the discussion of flux quantization of RR-fields not on Cauchy surfaces but on spacetime, a more intricate-looking duality constraint
has been considered (\cite[\S 3]{MooreWitten00}\cite[Def. 6]{DistlerFreedMoore11}).

\subsection{The C-field and Hypothesis H}
\label{TheCField}

On a spacetime of the form
$
  X^{11}
  \,=\,
  \mathbb{R}^{0,1}
  \times
  X^{10}
$
consider the 11d supergravity C-field with its canonical decomposition \eqref{CanonicalDecompositionOfFluxDensities}:
$$
  \begin{array}{l}
  G_4
  \;=\;
  B_4
    \,-\,
  \star
  E_7
  ,
  \hspace{1cm}
  \mbox{for}
  \left\lbrace
  \def\arraystretch{1.2}
  \begin{array}{l}
    B_4 \,\in\,
    \Omega^4_{\mathrm{dR}}(X^{11})_{\iota_{\partial_t} = 0 }
    \\
    E_7 \,\in\,
    \Omega^7_{\mathrm{dR}}(X^{11})_{\iota_{\partial_t} = 0 }
  \end{array}
  \right.
  \\[20pt]
  G_7
  \;=\;
  B_7
    \,-\,
  \star
  E_4
  ,
  \hspace{1cm}
  \mbox{for}
  \left\lbrace
  \def\arraystretch{1.2}
  \begin{array}{l}
    B_7 \,\in\,
    \Omega^7_{\mathrm{dR}}(X^{11})_{\iota_{\partial_t} = 0 }
    \\
    E_4 \,\in\,
    \Omega^4_{\mathrm{dR}}(X^{11})_{\iota_{\partial_t} = 0 }
    \,.
  \end{array}
  \right.
  \end{array}
$$
The higher Maxwell equations of 11d supergravity in their ``duality-symmetric'' form
\eqref{TheCovariantEquationsOfMotion} \cite{BandosBerkovitsSorokin98}\cite[\S 2]{CJLP98}\cite[\S 2]{BNS04} are:
\begin{equation}
  \def\arraystretch{1}
  \def\arraycolsep{20pt}
  \begin{array}{rcl}
  \differential
  \,
  G_4
  \,=\,
  0
  &
  \Leftrightarrow
  &
  \left\{\hspace{-6mm}
  \def\arraystretch{1.4}
  \begin{array}{c}
    \overset{
      \mathclap{
      \raisebox{3pt}{
        \scalebox{.7}{
          \color{gray}
          magnetic Gau{\ss} law
        }
      }
      }
    }{
      \differential_s \, B_4
      \,=\,
      0
    }
    \\
    \underset{
      \mathclap{
        \scalebox{.7}{
          \color{gray}
          magnetic evolution
        }
      }
    }{
    \differential_t \,  B_4
    \,=\,
    \differential_s \star E_7
    }
  \end{array}
  \right.
  \\[30pt]
  \differential
  \,
  G_7
  \,=\,
  -\tfrac{1}{2}G_4 \wedge G_4
  &\Leftrightarrow&
  \left\{\hspace{-6mm}
  \def\arraystretch{1.4}
  \begin{array}{c}
    \overset{
      \mathclap{
        \scalebox{.7}{
          \color{gray}
          electric Gau{\ss} law
        }
      }
    }{
    \differential_s
    \,
    B_7
    \,=\,
    -\tfrac{1}{2}
    B_4 \wedge B_4
    }
    \\
    \underset{
      \scalebox{.7}{
        \color{gray}
        electric evolution
      }
    }{
    \differential_t \,  B_7
    \,=\,
    \differential_s \star E_4
    +
    B_4 \wedge \star E_7\;.
    }
  \end{array}
  \right.
  \\[30pt]
  \star \, G_4 \,=\, G_7
  &\Leftrightarrow&
  \left\{
  \def\arraystretch{1.2}
  \begin{array}{l}
    E_4 \,=\, B_4
    \\
    \underset{
      \mathclap{
        \raisebox{-3pt}{
        \scalebox{.7}{
          \color{gray}
          \bf
          duality relation
        }
        }
      }
    }{
    B_7 \,=\, - E_7
    }
  \end{array}
  \right.
  \end{array}
\end{equation}
It may be instructive to explicitly check Lem. \ref{GaussLawIsFirstClassConstraint} in this case:
\[
  \def\arraystretch{1.4}
  \begin{array}{rl}
    \differential_t
    \,
    \big(
    \differential_s
    \,
    B_4
    \big)
    &
    =\,
    -
    \differential_s
    \,
    \big(
    \differential_t
    \,
    B_4
    \big)
    \\[-2pt]
  &  =\;
    -
    \differential_s
    \,
    \differential_s \star E_7
    \\[-2pt]
  &  =\;
    0\;,
  \\[2pt]
   \differential_t
   \,
   \big(
     \differential_s
     \,
     B_7
     +
     \tfrac{1}{2}
     B_4 \wedge B_4
   \big)
   &
   =\;
   -
   \differential_s
   \,
   \differential_t
   \,
   B_7
   +
   (\differential_t \, B_4)
   \wedge
   B_4
   \\
  & =\;
   -
   \differential_s
   \big(
     B_4 \wedge
     \star E_7
   \big)
   +
   (\differential_s \star E_7)
   \wedge
   B_4
   \\[-2pt]
  & =\; 0
   \,.
 \end{array}
\]
Remarkably, the solution space \eqref{CanonicalSolutionSpace} of the C-field
\begin{equation}
  \def\arraystretch{1.6}
  \begin{array}{rcl}
  \mathrm{SolSpace}_{
    \scalebox{.7}{C-Field}
  }
  &\simeq&
  \left\{
  \def\arraystretch{1.2}
  \begin{array}{l}
    B_4
      \,\in\,
    \Omega^4_{\mathrm{dR}}(X^{10})
    \\
    B_7 \,\in\,
    \Omega^7_{\mathrm{dR}}(X^{10})
  \end{array}
  \middle\vert
  \def\arraystretch{1.2}
  \begin{array}{l}
    \differential \, B_4 \,=\, 0
    \\
    \differential \, B_7 \,=\,
    -
    \tfrac{1}{2}
    B_4 \wedge B_4
  \end{array}
  \right\}
  \\
  &\simeq&
  \Omega_{\mathrm{dR}}\big(
    X^{10}
    ,\
    \mathfrak{l}S^4
  \big)_\closed
    \end{array}
\end{equation}
is characterized (in the sense of Rem. \ref{ModuliProblemOfSolutions}) by the ``M-theory gauge algebra'' \cite[(2.5)]{CJLP98}\cite[\S 4]{Sati10}\cite[\S 2.2]{SatiVoronov22b}
which happens to be
the Whitehead $L_\infty$-algebra of the 4-sphere (cf. \cite[Ex. 5.3]{Char}), an observation due to \cite[\S 2.5]{Sati13}, see also \cite[\S 2]{FSS17RationalSphere}:
$$
  \mathfrak{l}S^4
    =
  \mathbb{R}\langle
    v_3, \,v_6
  \rangle
  \quad
  \mbox{with}
  \quad
  [v_3, v_3] = v_6
  \,.
$$

It follows that the homotopy type of the 4-sphere is the classifying space for an admissible flux quantization law (Def. \ref{FluxQuantizationLaws}) of
the C-field on phase space (without a further duality constraint). The non-abelian cohomology theory classified by the 4-sphere is known as
4-{\it Cohomotopy}, i.e., unstable/nonabelian Cohomotopy, as in the original form of Pontrjagin's theorem (the origin of the now more famous
Pontrjagin-Thom theorem), identifying it with unstable framed cobordism, cf. \cite[\S 2.1 \& 2.2]{SS23Mf}.

\medskip
The evident hypothesis that this choice is the correct flux quantization law in M-theory is essentially what we called ``Hypothesis H'' \cite{FSS19TwistedCohomotopy}\cite{SS20Orientifold}\cite{FSS21HopfWZInHypothesisH}\cite{FSS21TwistedStringInHypothesisH}\cite{SS21M5Anomaly}\cite{FSS22TwistorialCohomotopy} \cite{SS23Mf}.
More precisely, Hypothesis H postulates C-field flux quantization by a {\it tangentially twisted} form of 4-Cohomotopy (not further discussed here for the sake of
brevity, but see the survey in \cite[\S 12]{Char}), which turns out to capture subtle topological effects such as the famous ``shifted'' flux quantization of
the C-field from \cite{Witten97a} on spacetimes whose tangent bundle has a non-trivial fractional Pontrjagin class $\tfrac{1}{2}p_1$ (\cite[Prop. 3.13]{FSS19TwistedCohomotopy},
a previously enigmatic phenomenon which was the key motivation in \cite{HopkinsSinger05InChargeQuantization} for introducing the notion of generalized
differential cohomology in the first place).

\medskip

The point that the present discussion adds to this picture is the observation that the corresponding moduli of (unstable) differential 4-Cohomotopy
\cite{FSS15}\cite[Ex. 9.3]{Char} of any Cauchy surface $X^{10}$ is already the phase space (Def. \ref{IntrinsicPhaseSpace}) of the cohomotopically
flux-quantized C-field, without the need to impose a further duality constraint:
\medskip
\begin{equation}
  \label{CohomotopicalPhaseSpace}
  \mathrm{PhaseSpace}^{\mathrm{FSS}}_{
    \scalebox{.7}{C-Field}
  }
  \;=\;
  \left\{
    C
      \,\in\,
    \overset{
      \mathclap{
        \raisebox{3pt}{
          \scalebox{.7}{
            \color{gray}
            \bf
            canonical differential 4-cohomotopy
          }
        }
      }
    }{
    \Omega_{\mathrm{dR}}\big(
      X^{10}
      ;\,
      \mathfrak{l}S^4
    \big)_{\closed}
      \quad \;
      \underset{
        \mathclap{ \!\!\!\!\!\!\!\!\!\!
          L^{\mathbb{R}}S^4(X^{10})
        }
        \hspace{-20pt}
      }{\times}
      \quad \;
      S^4(X^{10})
      }
  \right\}
  \,.
\end{equation}

\smallskip

In fact, in this case, at least, this flux quantization naturaly lifts form the phase space to all of spacetime, or rather to {\it super}-spacetime: This is the content of \cite{GGS24}.

\begin{remark}[Double dimensional reduction of the C-field]
  \label{DoubleDimensionalReductionOfTheCField}
  Many examples of higher gauge fields naturally arise via double dimensional KK-reduction (reducing both the spacetime dimension as well as higher gauge field degrees) of examples in higher dimensions, and specifically of the C-field in 11-dimensions (cf. \cite[Ex. 2.5]{SS24Flux}). This is traditionally discussed at the level of differential forms, but at least if the fiber space $F$ has group structure then double dimensional reduction applies also to flux quantization laws $\mathcal{A}$, namely by forming the ``cyclification'' $\mathrm{Cyc}_F(\mathcal{A}) \,:=\, \mathrm{Map}(F,A)\sslash F$ of the classifying space $\mathcal{A}$ \cite[\S 2.2]{BMSS19}\cite{SS24Cyc}. In this fashion flux quantization laws of all the previous examples are induced by flux quantization laws of the C-field. Since on the underlying flux densities this process produces the usual U-duality groups \cite{SatiVoronov22b}\cite{SV23RHT}, it follows that cyclification of flux quantization laws necessarily yields a form of U-duality covariant flux quantization. This deserves to be further discussed elsewhere.
\end{remark}

\medskip

\noindent
{\bf Outlook -- Fundamental M-branes from cohomotopical C-field flux quantization.}
For spacetimes modelling the vicinity of flat solitonic branes of small codimension, the cohomotopical phase space \eqref{CohomotopicalPhaseSpace}
of the C-field is shape-equivalent\footnote{A map of smooth $\infty$-stacks is a {\it shape-equivalence} if it induces a weak homotopy equivalence
under passage to topological realizations (cf. \cite[Prop. 1.26]{Char}\cite[Def. 3.1]{SS20Orbi}).} to a configuration space manifold of
points -- namely of the positions of these branes in their transverse space \cite[\S 2]{SS22Config}, which serves as an atlas for the phase
space stack. With topological quantum observables of flux-quantized phase spaces conceived as in the companion article \cite{SS23QuantumFluxObservables},
Hypothesis H hence implies a rich structure of quantum observables on at least certain M-brane configurations. In \cite[\S 4]{SS22Config}
these quantum observables are matched to a wealth of structures expected in the string/M-theory literature, in particular to quantum states
of transverse M2/M5-branes as seen in the BMN matrix model \cite[\S 4.9]{SS23QuantumFluxObservables}\cite{CSS23Chord}.

\smallskip
A question that had been left open in \cite{SS23QuantumFluxObservables} is whether (or why not) the purely cohomotopical structures considered
there would (not) need to be subjected to a metric duality constraint in order to become ``physical''. We suggest that the result presented here
may be understood as answering this question (cf. also \cite{GGS24}).

\newpage

\appendix

\section{Appendix: Computations}
\label{Computations}

The proof of Lem. \ref{GaussLawIsFirstClassConstraint} is quite elementary and immediate --- thanks to the duality-symmetric form of
the higher Maxwell equations \eqref{TheCovariantEquationsOfMotion} --- but for the record we spell it out in detail:

\medskip

\noindent The key point is that with the canonical decomposition \eqref{CanonicalDecompositionOfFluxDensities}
\[
  F^{(i)}
  \,=\,
  B^{(i)} \,+\, \star E^{(i)}
  \,,
  \hspace{1cm}
  \mbox{for}
  \left\{
  \def\arraystretch{1.5}
  \begin{array}{l}
    B^{(i)} \,\in\,
    \Omega^{\mathrm{deg}_i}_{\mathrm{dR}}\big(
      X^D
    \big)_{\iota_{\partial_t} = 0}
    \\
    E^{(i)} \,\in\,
    \Omega^{D-\mathrm{deg}_i}_{\mathrm{dR}}\big(
      X^D
    \big)_{\iota_{\partial_t} = 0}
  \end{array}
  \right.
\]
the terms $\star \vec E$ have vanishing wedge product with each other since they are all proportional to differential $\differential t$
of the temporal coordinate function globally given by \eqref{FoliationByCauchySurfaces}.

Therefore, any graded-symmetric polynomial function $\vec P$ of $\vec F$ is the sum of that same polynomial function of $\vec B$ with
the result of
iteratively replacing in this polynomial every one factor of $B^{(i)}$ by $\star E^{(i)}$.
The latter may be written as the sum of graded partial derivatives $\frac{\delta}{\delta B^{(i)}}$ of the polynomial times $\star E^{(i)}$:
\begin{equation}
  \label{ExpandingOutPolynomial}
  \vec P\big(
    \vec F
  \big)
  \;=\;
  \vec P\big(
    \vec B
  \big)
  \,+\,
  \big(
    \star E^{(i)}
  \big)
  \wedge
  \frac{\delta}{\delta B^{(i)}}
  \vec P\big(
    \vec B
  \big)
  \,.
\end{equation}
For example, in the situation of \cref{TheCField}:
\[
  \tfrac{1}{2}
  F_4 \wedge F_4
  \;=\;
  \tfrac{1}{2}
  \big(
    B_4 + \star E_7
  \big)
  \wedge
  \big(
    B_4 + \star E_7
  \big)
  \;=\;
  \tfrac{1}{2} B_4 \wedge B_4
  \,+\,
  (\star E_7) \wedge B_4
  \,.
\]

\medskip

Similarly, in the differential of $\vec F$ we collect terms with and without a factor of $\differential t$:
\begin{equation}
  \label{DifferentialF}
    \differential \, \vec F
    \;=\;
    \big(
      \differential_s
      +
      \differential_t
    \big)
    \big(
      \vec B \,+\, \star \vec E
    \big)
    \;=\;
      \differential_s \, B
    \,+\,
    \big(
      \differential_t \, B
      \,+\,
      \differential_s \, \star E
    \big)
    \,.
\end{equation}
Now equating \eqref{DifferentialF} with \eqref{ExpandingOutPolynomial} and matching temporal/nontemporal components
already yields the Gau{\ss}- and the Faraday-Amp{\`e}re law in \eqref{GaussFaradayAmpereEquations}.
Similarly for their integrability condition:
\begin{equation}
  \label{GettingTheIntegrabilityCondition}
  \def\arraystretch{1.6}
  \begin{array}{l}
    \mathllap{
      0 \,=\;\,
    }
    \differential
    \,
    \vec P\big(
      \vec F
    \big)
    \\
    \;=\;
    \big(
      \differential_s
      \,+\,
      \differential_t
    \big)
    \Big(
      \vec P\big( \vec B \big)
      \,+\,
      \big(\star E^{(i)}\big)
      \wedge
      \frac{\delta}{\delta B^{(i)}}
      \vec P\big(\vec B\big)
    \Big)
    \\
    \;=\;
    \differential_s
    \,
    \vec P\big(
      \vec B
    \big)
    \,+\,
    \bigg(
      \differential_t
      \,
      \vec P\big(
        \vec B
      \big)
      \,+\,
      \differential_s
      \Big(
      \big(\star E^{(i)}\big)
      \wedge
      \frac{\delta}{\delta B^{(i)}}
      \vec P\big(\vec B\big)
      \Big)
    \bigg).
  \end{array}
\end{equation}

\medskip

With this, we may spell out the derivation in the proof of Lem. \ref{GaussLawIsFirstClassConstraint}:
\begin{equation}
  \def\arraystretch{1.8}
  \begin{array}{ll}
    \differential_t
    \Big(
      \differential_s
      \,
      \vec B
      -
      \vec P\big(\vec B\big)
    \Big)
    \\
    \;=\;
    -
    \differential_s
    \,
    \differential_t
    \,
    \vec B
    \,-\,
    \differential_t\,
    \vec P\big(
      \vec B
    \big)
    \\
    \;=\;
    -\differential_s
    \Big(
      -
      \differential_s \star E
      +
      \big(
        \star E^{(i)}
      \big)
      \wedge
      \frac{\delta}{\delta B^{(i)}}
      \vec P\big(
        \vec B
      \big)
    \Big)
    \,+\,
    \differential_s
    \Big(
      \big(
        \star E^{(i)}
      \big)
      \wedge
      \frac{\delta}{\delta B^{(i)}}
      \vec P\big(
        \vec B
      \big)
    \Big)
    \\
    \;=\;
    \underbrace{
    \differential_s \,
    \differential_s
    \,
    \star E
    \mathclap{
      \phantom{\vert_{\vert_{\vert_{\vert}}}}
    }
    }_{= 0}
    \;
    \underbrace{
    \,-\;
    \differential_s
    \Big(
      \big(
        \star E^{(i)}
      \big)
      \wedge
      \textstyle{
        \frac{\delta}{\delta B^{(i)}}
      }
      \vec P\big(
        \vec B
      \big)
    \Big)
      +
    \differential_s
    \Big(
      \big(
        \star E^{(i)}
      \big)
      \wedge
      \textstyle{
        \frac{\delta}{\delta B^{(i)}}
      }
      \vec P\big(
        \vec B
      \big)
      \Big)
    }_{ = 0 }
    \\
    \;=\
    0
    \,.
  \end{array}
\end{equation}
Here the second line uses \eqref{BicomplexStructure}, the third line uses \eqref{ExpandingOutPolynomial},
\eqref{DifferentialF}, \& \eqref{GettingTheIntegrabilityCondition}, and in the last line we use \eqref{BicomplexStructure} and cancel summands.

\end{document}